\documentclass[12pt,leqno,letterpaper]{article}

\usepackage{amsmath,amsthm,enumerate,amssymb}
\usepackage[latin1]{inputenc}
\usepackage[T1]{fontenc}
\usepackage[english]{babel}
\usepackage{graphicx}

\newtheorem{theorem}{Theorem}[section]

\newtheorem{definition}[theorem]{Definition}

\newcommand{\tr}{{\rm Tr\hskip -0.2em}~}

\begin{document}

\title{Convexity of quantum $\chi^2$-divergence}
\author{Frank Hansen}
\date{February 15 2011}

\maketitle

\begin{abstract}

The quantum $ \chi^2 $-divergence has recently been introduced
and applied to quantum channels (quantum Markov processes).  In contrast to the classical setting the 
quantum $ \chi^2 $-divergence is not unique but depends on the choice of quantum statistics.
In the reference \cite{kn:ruskai:2010:1} a special one-parameter family of quantum  $ \chi^2_\alpha(\rho,\sigma) $-divergences for density matrices were studied, and it was established that they are convex functions in $ (\rho,\sigma) $ for parameter values $ \alpha\in [0,1], $ thus mirroring the classical theorem for the   $ \chi^2(p,q) $-divergence for probability distributions $ (p,q). $ We prove that any quantum $ \chi^2 $-divergence is a convex function in its two arguments.\\[1ex]

\end{abstract}

\section{Introduction}

The geometrical formulation of quantum statistics originates in a study by Chent\-sov of the classical
Fisher information.
Chentsov proved \cite{kn:censov:1982} that the Fisher-Rao metric is the
only Riemannian metric, defined on the tangent space, that is decreasing under Markov morphisms.
    Since Markov morphisms represent coarse graining or randomization, it
    means that the Fisher information is the only Riemannian metric
    possessing the attractive property that distinguishability of probability
    distributions becomes more difficult when they are observed through a
    noisy channel.

    Morozova \cite{kn:morozova:1990} extended the analysis to quantum mechanics by
    replacing Riemannian metrics defined on the tangent space of the simplex
    of probability distributions with positive definite sesquilinear
    (originally bilinear) forms $ K_\rho $ defined on the tangent space of a
    quantum system, where $ \rho $ is a positive definite state. Customarily,
    $ K_\rho $ is extended to all operators (matrices) supported by the
    underlying Hilbert space, cf. \cite{kn:petz:1996:2,kn:hansen:2006:2} for
    details. Noisy channels are in this setting represented by stochastic
    (completely positive and trace preserving) mappings,  and the
    contraction property is replaced by the monotonicity requirement
    \[
    K_{T(\rho)}(T(A),T(A))\le K_\rho(A,A)
    \]
    for every stochastic
    mapping $ T:M_n(\mathbb C)\to M_m(\mathbb C). $ Unlike the classical
    situation, these requirements no longer uniquely determine the metric.
    
    We consider the following class of functions which is used to characterize monotone metrics.
    
      \begin{definition}
$ {\cal F}_{\text{op}}$ is the class of functions $f: (0,+\infty) \to (0,+\infty)$ such that

\begin{enumerate}[(i)]
\item $f$ is operator monotone,
\item $f(t)=tf(t^{-1})$ for $ t>0, $
\item $f(1)=1.$
\end{enumerate}
\end{definition}
    
    By the combined efforts of Chentsov, Morozova and Petz  \cite{kn:petz:1996:2}  it was established that a monotone metric is given on the canonical form
\begin{equation}\label{canonical form of monotone metric}
 K_\rho(A,B)=\tr A^* c(L_\rho,R_\rho) (B),
\end{equation}
 where the so-called Morozova-Chentsov function $ c $ is of the form
    \[
    c(x,y)=\frac{1}{y f(xy^{-1})}\qquad f\colon\mathbf R_+\to\mathbf R_+
    \]
   for a function $ f\in{\cal F}_{\text{op}} $ and $ c $ is taken in the two commuting positive definite (super) operators $ L_\rho $ and $ R_\rho $ defined by setting
  \[
  L_\rho A=\rho A\quad\text{and}\quad R_\rho A=A\rho.
  \]
    It is condition $ (ii) $ in the definition above that ensures symmetry of the metric in the sense that $    K_\rho(A,B)=  K_\rho(B,A) $ for self-adjoint $ A $ and $B. $
   
      Lesniewski and Ruskai \cite{kn:ruskai:1999} gave equivalent descriptions in terms of operator convex functions and operator monotone decreasing functions. In particular
\[
 K_\rho(A,B)=\tr A^* R_\rho^{-1}k(L_\rho R_\rho^{-1}) B,      
\]
where $ k\colon\mathbf R_+\to\mathbf R_+ $ is an operator monotone decreasing function satisfying $ k(t^{-1})=tk(t) $ for $ t>0 $ and $ k(1)=1. $ The Morozova-Chentsov-Petz formalism is then recovered by inserting the operator monotone function $ f(t)=1/k(t) $ in equation (\ref{canonical form of monotone metric}).

\begin{definition}
The $ \chi^2_f $-divergence (relative to a choice of monotone metric)  is given by
\[
\chi^2_f(\rho,\sigma)=K_\sigma^c(\rho-\sigma,\rho-\sigma),
\]
where 
\[
 c(x,y)=\frac{1}{yf(xy^{-1})} \qquad x,y>0
\]
is the Morozova-Chentsov function specified by a function $ f\in  {\cal F}_{\text{op}}\,. $
\end{definition}

The functions
\[
f_\alpha(t)=\frac{2 t^\alpha}{t^{2\alpha-1}+1}\qquad t>0
\]
with parameter $ \alpha\in [0,1] $ are elements in  $  {\cal F}_{\text{op}} $ and correspond to the functions
\[
k(t)=\frac{1}{2}(t^{-\alpha}+t^{\alpha-1})\qquad t>0
\]
in the Ruskai-Lesniewski formalism. The associated Morozova-Chentsov functions are given by
\[
c_\alpha(x,y)=\frac{1}{yf_\alpha(xy^{-1})} =\frac{x^{\alpha-1}y^{-\alpha}+x^{-\alpha} y^{\alpha-1}}{2}
\]
and the divergences for positive definite density matrices are
\[
\chi^2_\alpha(\rho,\sigma)=\tr(\rho-\sigma)\sigma^{-\alpha}(\rho-\sigma)\sigma^{\alpha-1}
=\tr \rho\sigma^{-\alpha}\rho\sigma^{\alpha-1}-1.
\]
This expression is convex in $ (\rho,\sigma) $ as pointed out in \cite{kn:ruskai:2010:1}.

\section{Convexity}

Consider a function $ f\colon \mathbf R_+\to\mathbf R_+. $ The perspective (function) of $ f $ is the function $ g $ of two variables defined by setting
\[
g(t,s)=sf(ts^{-1})\qquad t,s>0.
\]
Effros showed \cite{kn:effros:2009:1} that if $ f $ is operator convex then, whenever meaningful, the perspective is a convex operator function of two variables. In particular, it becomes operator convex in the sense of Kor{\'a}nyi \cite{kn:koranyi:1961}. It is also convex in functions of commuting operators. In particular, the function
\[
(\rho,\sigma)\to g(L_\rho,R_\sigma)
\]
is convex in pairs of positive definite $ n\times n $ matrices,
equivalent to the statement that the function
\[
(\rho,\sigma)\to \tr A^* g(L_\rho,R_\sigma)A
\]
is convex for any $ n\times n $ matrix $ A. $ Similar statements are valid also for operator concave functions.
 
\begin{theorem}
The $ \chi^2_f(\rho,\sigma) $-divergence is convex in  $ (\rho,\sigma) $ for any  $ f $ in $ {\cal F}_{op}\,.  $
\end{theorem}

\begin{proof}
Let us consider a function  $ f $ in $ {\cal F}_{op}  $ with Morozova-Chensov function
\[
c(x,y)=\frac{1}{yf(xy^{-1})}=F(x,y)^{-1}\qquad x,y>0,
\]
where $ F(x,y)=yf(xy^{-1}) $ is the perspective of  $ f. $ Since $ f $ is operator concave we obtain that $ F(x,y) $ is operator concave as a function of two variables.  Inversion (of super operators) is decreasing. By using linearity of the mappings $ \sigma\to L_\sigma $ and $ \sigma\to R_\sigma $ we therefore obtain
\[
\begin{array}{l}
F(L_{\lambda\sigma_1+(1-\lambda)\sigma_2},R_{\lambda_1\sigma_1+(1-\lambda)\sigma_2})^{-1}\\[2ex]
=F(\lambda L_{\sigma_1}+(1-\lambda)L_{\sigma_2},\lambda R_{\sigma_1}+(1-\lambda)R_{\sigma_2})^{-1}\\[2ex]
\le \bigl(\lambda F(L_{\sigma_1},R_{\sigma_1})+(1-\lambda)F(L_{\sigma_2},R_{\sigma_2})\bigr)^{-1}
\end{array}
\]
for states (density matrices) $ \rho_1,\rho_2 $ and $ \sigma_1,\sigma_2 $ and real numbers $ \lambda\in [0,1]. $  
The divergence $ \chi^2_f(\rho,\sigma) $ is given on the form
\[
\chi^2_f(\rho,\sigma)=\tr (\rho-\sigma) c(L_\sigma,R_\sigma)(\rho-\sigma),
\]
and by setting
\[
 \rho=\lambda\rho_1+(1-\lambda)\rho_2\quad\text{and}\quad \sigma=\lambda\sigma_1+(1-\lambda)\sigma_2 
 \]
we obtain the inequality
\[
\begin{array}{l}
\chi^2_f(\rho,\sigma)=\tr (\rho-\sigma) F(L_\sigma,R_\sigma)^{-1}(\rho-\sigma) \\[2ex]
=\tr (\rho-\sigma) F(L_{\lambda \sigma_1+(1-\lambda)\sigma_2},R_{\lambda \sigma_1+(1-\lambda)\sigma_2})^{-1}(\rho-\sigma),\\[2ex]
\le \tr (\rho-\sigma)\bigl(\lambda F(L_{\sigma_1},R_{\sigma_1})+(1-\lambda) F(L_{\sigma_2},R_{\sigma_2})\bigr)^{-1}(\rho-\sigma)\\[2ex]
=\tr\bigl(\lambda(\rho_1-\sigma_1)+(1-\lambda)(\rho_2-\sigma_2)\bigr)\bigl(\lambda F(L_{\sigma_1},R_{\sigma_1})+(1-\lambda) F(L_{\sigma_2},R_{\sigma_2})\bigr)^{-1}\\[1ex]
\hfill\bigl(\lambda(\rho_1-\sigma_1)+(1-\lambda)(\rho_2-\sigma_2)\bigr)\\[2ex]
\le\lambda\tr (\rho_1-\sigma_1) F(L_{\sigma_1},R_{\sigma_1})^{-1} (\rho_1-\sigma_1)+ (1-\lambda)\tr (\rho_2-\sigma_2) F(L_{\sigma_2},R_{\sigma_2})^{-1} (\rho_2-\sigma_2)\\[2ex]
=\lambda\chi^2_f(\rho_1,\sigma_1)+(1-\lambda)\chi^2_f(\rho_2,\sigma_2),
\end{array}\]
where we in the second inequality, applied on super operators, used that the mapping 
\[
(A,\xi)\to (\xi\mid A^{-1}\xi)
\]
is jointly convex for positive invertible operators $ A $ on a Hilbert space $ H, $ and vectors $ \xi\in H, $
 cf. \cite[Proposition 4.3]{kn:hansen:2006:3}. This is also a direct consequence of convexity of the mapping
\[
(A,B)\to B^* A^{-1}B 
\]
for $ B $  arbitrary and $ A $  positive definite, cf.~\cite[Remark after Theorem 1]{kn:lieb:1974} and \cite[Remark 4.5]{kn:hansen:2006:3}. Furthermore, it is related to \cite[Theorem 3.1]{kn:ando:2010:1}.
\end{proof}

Any function $ f $ in  $ {\cal F}_{op}  $ satisfies the inequalities
\[
\frac{2t}{t+1}\le f(t) \le \frac{t+1}{2}\qquad t>0,
\]
where the smallest function in $ {\cal F}_{op}  $ corresponds to the Bures metric. We mention the following characterization \cite{kn:hansen:2006:2,kn:hansen:2008:1,kn:audenaert:2008} of the functions in $ {\cal F}_{op} . $ 

\begin{theorem}
    A function $ f $ in $ {\cal F}_{op}  $ admits a canonical representation
    \begin{equation}\label{canonical representation of f}
    f(t)=\frac{1+t}{2}\exp\left[-\int_0^1\frac{(1-\lambda^2)(1-t)^2}{(\lambda+t)(1+\lambda t)(1+\lambda)^2}\,h(\lambda)\,d\lambda\right],
    \end{equation}
    where the weight function $ h:[0,1]\to[0,1] $ is measurable.
    The equivalence class containing $ h $ is uniquely determined by $ f. $ Any function
    on the given form is in $ {\cal F}_{op}.  $
   \end{theorem}
   
   Notice that the integral kernel is non-negative for every $ t>0. $ The representation induces an order relation $ \preceq $ in  $ {\cal F}_{op}  $ stronger than the point-wise order by setting $ f\preceq g $ if the representing weight functions $ h_f $ and $ h_g $ satisfies $ h_f \ge h_g $ almost everywhere. With this order relation $ ({\cal F}_{op}\,, \preceq)  $ becomes a lattice, inducing a lattice structure on the set of quantum $ \chi^2 $ divergences. It is compatible with the parametrization of the Wigner-Yanase-Dyson metrics, cf. \cite[Theorem 2.8]{kn:audenaert:2008}. 
   
   The representation in (\ref{canonical representation of f}) may be used to construct families of metrics that increase monotonously from the smallest (Bures) metric to the largest. In fact any family of weight functions that decreases monotonously from the constant  $1$ to the zero function will induce this property. In \cite[Proposition 3]{kn:hansen:2006:2} we considered the constant weight functions $ h_\alpha=\alpha $ for $ 0\le\alpha\le 1 $ and obtained in this way a family of metrics
\[
f_\alpha(t)=t^\alpha\left(\frac{1+t}{2}\right)^{1-2\alpha}\qquad t>0,
\]
that decreases monotonously from the largest monotone metric down to the Bures metric for $ \alpha $ increasing from $ 0 $ to $ 1. $ In the Lesniewski-Ruskai representation that corresponds to the functions
\[
k_\alpha(\omega)=w^{-\alpha}\left(\frac{1+\omega}{2}\right)^{2\alpha-1}
\]
as mentioned in \cite[Appendix A]{kn:ruskai:2010:1}.



\vfill

      \noindent Frank Hansen: Institute for International Education, Tohoku University, Japan. Email: frank.hansen@m.tohoku.ac.jp.

\end{document}